\documentclass[12pt,a4paper]{article}
\usepackage[latin1,utf8]{inputenc}
\usepackage{amsmath}
\usepackage{amsthm}
\usepackage{amsfonts}
\usepackage{amssymb}
\usepackage{thmtools}
\usepackage{thm-restate}
\usepackage{hyperref}
\usepackage{cleveref}
\usepackage{graphicx}
\usepackage{nicefrac}
\usepackage{fullpage}
\usepackage{bm}
\usepackage[linesnumbered,ruled]{algorithm2e}
\usepackage{tikz}
\usepackage{enumitem}
\usepackage{mathtools}
\usepackage{wasysym}

\DeclareMathOperator{\E}{\mathbb{E}}
\DeclareMathOperator{\poly}{\textup{poly}}
\newcommand\ceil[1]{\left\lceil{#1}\right\rceil}
\newcommand\N{\mathbb{N}}
\newcommand\floor[1]{\left\lfloor{#1}\right\rfloor}

\newcommand\PP{\mathcal{P}} % The input distribution to both Mitzenmacher and Drinea and to the Blahut Arimoto algorithm.
\newcommand\RR{\mathcal{R}} % The rate of the BAA output.
\newcommand\II{\mathcal{X}} % The input alphabet of the BAA
\newcommand\OO{\mathcal{Y}} % The output alphabet of the BAA
\newcommand\QQ{\mathcal{Q}} % The inner variable of the BAA
\newcommand\ii{x}           % An iterator over letters in the input alphabet of the BAA
\newcommand\oo{y}           % An iterator over letters in the output alphabet of the BAA
\newcommand\BAADenominator{\mathcal{D}}
\newcommand\BAAWeight{\mathcal{W}}

\newtheorem{theorem}{Theorem}[section]
\newtheorem{thm}{Theorem}[section]

\newtheorem{lemma}[thm]{Lemma}

\newtheorem{claim}[thm]{Claim}

\newtheorem{definition}{Definition}

\newcommand\R{\mathbb{R}}

\newcommand\defeq{\stackrel{def}{=}}

\newcommand\abs[1]{{\left\lvert{#1}\right\rvert}}

\usepackage{soul}
\usepackage{xcolor}

\usepackage{mwe}
\usepackage{float}
\usepackage[font=small]{caption}
\numberwithin{equation}{section}

\usepackage{pgfplots}

\usetikzlibrary{calc,positioning,shapes,backgrounds, decorations.pathreplacing}

\newcommand{\cC}{\mathcal{C}}

\newcommand{\zo}{\{0,1\}}
\newcommand{\rate}{\text{Rate}}
\newcommand{\ourBDCrate}{0.1221}
\newcommand{\ourBDCUpperBound}{0.3745}
\newcommand{\delForUB}{0.68}
\newcommand{\Rplus}{\R_{+}}

\usepackage{authblk}

\begin{document}

	\title{Improved Upper and Lower Bounds on the Capacity of the Binary Deletion Channel}

	\author[1]{Ittai Rubinstein\thanks{\href{mailto:ittai.rubinstein@gmail.com}{ittai.rubinstein@gmail.com}}}
	\author[2]{Roni Con\thanks{\href{mailto:roni.con93@gmail.com}{roni.con93@gmail.com}\\The work of Roni Con was partially supported by the European Research Council (ERC grant number 852953) and by the Israel Science Foundation (ISF grant number 1030/15).}}
	\affil[1]{Qedma Quantum Computing, Tel-Aviv, Israel}
	\affil[2]{Blavatnik School of Computer Science, Tel Aviv University, Tel-Aviv, Israel}
	\date{\today}
	\maketitle
	\begin{abstract}

	The {\em binary deletion channel} with deletion probability $d$ ($\textup{BDC}_d$) is a random channel
	that deletes each bit of the input message i.i.d with probability $d$.
	It has been studied extensively as a canonical example of a channel with synchronization errors~\cite{mitzenmacher2009survey, mercier2010survey, cheraghchi2020overview}.

	Perhaps the most important question regarding the BDC is determining its capacity.
	Mitzenmacher and Drinea~\cite{mitzenmacher2006simple} and Kirsch and Drinea~\cite{kirsch2009directly} show a method by which distributions on run lengths can be converted to codes for the BDC, yielding a lower bound of $\cC(\textup{BDC}_d) > 0.1185 \cdot (1-d)$.
	Fertonani and Duman~\cite{fertonani2010novel}, Dalai~\cite{dalai2011new} and Rahmati and Duman~\cite{rahmati2014upper} use computer aided analyses based on the Blahut-Arimoto algorithm to prove an upper bound of $\cC(\textup{BDC}_d) < 0.4143\cdot(1-d)$ in the high deletion probability regime ($d > 0.65$). 

	In this paper, we show that the Blahut-Arimoto algorithm can be implemented with a lower space complexity, allowing us to extend the upper bound analyses, and prove an upper bound of $\cC(\textup{BDC}_d) < \ourBDCUpperBound \cdot(1-d)$ for all $d \geq \delForUB$.
	Furthermore, we show that an extension of the Blahut-Arimoto algorithm can be used to select better run length distributions for Mitzenmacher and Drinea's construction, yielding a lower bound of $\cC(\textup{BDC}_d) > \ourBDCrate \cdot (1 - d)$.

	\end{abstract}
	\newpage
	\tableofcontents
	\newpage
	\section{Introduction}\label{sec:introduction}
		
In this paper we focus on the {\em binary deletion channel} (BDC) which deletes each bit from the input message randomly and independently with a given {\em deletion probability} $d$.
Loosely speaking, a channel is a medium over which messages are sent. A channel is defined by the way in which it introduces errors to the transmitted messages (also called codewords when they come from an error correcting code).

Two of the most well-studied channels are the Binary Erasure Channel (BEC$_d$) where each bit is independently replaced by a question mark with probability $d$ and the Binary Symmetric Channel (BSC$_d$) where each bit is independently flipped with probability $d$.
We note that the BDC is very different from the BEC and the BSC.
For example, consider the case where the transmitted message was $1110101$ and corruptions occurred in locations $2$ and $5$.
In this scenario, the BEC will return the word $1?10?01$ and the BSC will return $1010001$, while the BDC would return $11001$.

In other words, while the BEC and the BSC may corrupt some bits in the message, the BDC can change the length of a codeword and the index in which a bit may appear.
This makes the BDC a ``synchronization channel'', and significantly complicates its analysis. In fact, one of the main reasons for introducing the BDC was to model synchronization errors in communication. More recently, codes correcting from deletions found applications in a variety of fields such as computational biology and DNA storage~\cite{bornholt2016dna,yazdi2017portable,heckel2019characterization}.
For a more detailed review of synchronization channels and their applications, we refer the interested reader to the excellent surveys by Mitzenmacher~\cite{mitzenmacher2009survey}, Mercier et al.~\cite{mercier2010survey}, and Cheraghchi and Ribeiro~\cite{cheraghchi2020overview}.

To explain the question that we study we need some basic notions from coding theory. Recall that a binary error correcting code can be described either as an encoding map $C : \zo^k \rightarrow \zo^n$ or, abusing notation, as the image of such a map $C$. 
The rate of such a code $C$ is $\rate(C)=k/n$, which intuitively captures the amount of information encoded in every bit of a codeword. Naturally, we would like the rate to be as large as possible, but there is a tension between the rate of the code and the amount of errors/noise it can tolerate.

One of the most fundamental questions when studying a channel is to determine its capacity, i.e., the maximum achievable transmission rate over the channel that still allows recovering from the errors introduced by the channel, with high probability. Shannon proved in his seminal work~\cite{shannon1948mathematical} that the capacity of the BSC$_d$ is $1-h(d)$, where $h(\cdot)$ is the binary entropy function (for $0<x<1$, $h(x)=-x\log x-(1-x)\log{1-x}$).\footnote{All logarithms in this paper are base $2$.}
Elias~\cite{elias}, who introduced the BEC$_d$, proved that its capacity is $1-d$.

What about the capacity of the BDC$_d$? In spite of significant efforts by many researchers, much less is known about the capacity of the BDC.
This is because the asynchronous nature of the BDC which makes it interesting also makes it harder to analyse.

In the extremal parameter regimes, the behavior of the capacity of the BDC is partially understood.
When $d\rightarrow 0$ the capacity approaches $1-h(d)$~\cite{kalai2010tight}.
For all $d\in(0,1)$, the capacity is bounded from below by $\cC(\textup{BDC}_d) > 0.1185 \cdot (1-d)$~\cite{mitzenmacher2006simple}, and when $d > 0.65$, it is bounded from above by $\cC(\textup{BDC}_d) < 0.4143 \cdot (1- d)$~\cite{fertonani2010novel,dalai2011new,rahmati2014upper}, giving us the asymptotic scaling as $d\rightarrow 1$.
For a more detailed picture of the known bounds on the capacity of the BDC, see Figure~\ref{fig:fixed_rates_vs_general}.

\subsection{Previous Work}\label{subsec:previous-work}
Our focus in this paper is in bounding the capacity of the BDC in either the bulk of the parameter regime $d\in(\varepsilon, 1-\varepsilon)$, or in the high-deletion probability regime $d\rightarrow 1$.
We will list here the best known results for these regimes.

\paragraph{Lower bounds}
The best known lower bound on the capacity is due to  Mitzenmacher and Drinea~\cite{mitzenmacher2006simple, drinea2007improved} who showed a lower bound of $0.1185 \cdot (1-d)$ for all $d$, meaning that there are codes of this rate such that every transmitted codeword is decoded correctly with high probability.
This lower bound is the best known lower bound for the high deletion probability regime ($d\geq 0.95$).
For smaller values of $d$, Drinea and Mitzenmacher prove stronger bounds in~\cite[Table 1]{drinea2007improved}.
Their proof is constructive, but it does not directly yield an efficient decoding algorithm for the family of codes they construct.

Since then, several constructions of efficiently decodable codes for the BDC have been published~\cite{guruswami2017efficiently,con2019explicit}, culminating in several constructions of efficiently decodable codes that achieve capacity~\cite{tal2021polar,rubinstein2021explicit,pernice2022efficient}.
In particular,~\cite{rubinstein2021explicit} presents a method for converting any code for the BDC to an efficiently decodable one with an arbitrarily close rate.
Therefore, the codes generated using Mitzenmacher and Drinea's construction can be converted to explicit and efficient constructions of high rate codes for the BDC.

\paragraph{Upper bounds}
Fertonani and Duman~\cite{fertonani2010novel} proved several upper bounds on the capacity of the BDC.
They do this by providing the transmitter and the receiver with ``hints'' about the noise of the channel.
Adding these hints only increases the information rate, allowing them to bound the capacity of the BDC from above by bounding the capacity of some auxiliary channels using the Blahut-Arimoto algorithm.

Dalai~\cite{dalai2011new} and Rahmati and Duman~\cite{rahmati2014upper} refine this analysis and prove that for any $d\in (0.65, 1)$, the capacity of the BDC is bounded from above by $\cC(\textup{BDC}_d) \leq 0.4143(1-d)$.
Given unlimited computational resources, these methods will converge to the capacity of the channel, but this convergence is extremely slow.

\subsection{Our results}\label{subsec:our-results}
In this work, we improve both the upper and the lower bounds on the capacity of the BDC.
To improve the upper bound, we integrate into the classical Blahut-Arimoto algorithm several ingredients that reduce its memory requirements.
These ingredients include loop nest optimization, caching, symmetries, and compressing sparse matrices.
This allows us to extend the previous analyses to more computationally challenging regimes and gives us the following new upper bound
\begin{theorem} \label{thm:upper-bound-thm}
	For any $d \geq \delForUB$, it holds that
	\[
	\cC(\textup{BDC}_d) \leq \ourBDCUpperBound \cdot (1-d)\;.
	\]
\end{theorem}

Our second contribution is an improved lower bound on the capacity of the BDC based on Mitzenmacher and Drinea's construction~\cite{drinea2007improved}.
Mitzenmacher and Drinea's construction is parametrized by a distribution on ``run lengths'' and outputs a provable lower bound on the capacity of the BDC.
Mitzenmacher and Drinea apply their construction to various geometric run length distributions, but give no argument for why geometric distributions should be optimal.
We propose a heuristic approach to selecting better run length distributions to be used as parameters for Mitzenmacher and Drinea's rigorous analysis, yielding an improved lower bound (Theorem~\ref{thm:lower-bound}).

\begin{theorem}
	\label{thm:lower-bound}
	For any $d\in (0,1)$
	\[\cC(\textup{BDC}_d) > \ourBDCrate \cdot (1-d)\;.\]
\end{theorem}
Using the result of~\cite{rubinstein2021explicit} as a black box, we can efficiently construct a binary code that achieves this rate and has efficient encoding and decoding algorithms.

For a full characterization of the lower and upper bounds on the capacity of the BDC as a function of $d$, see Tables~\ref{tab:upper-bound} and~\ref{tab:lower-bound}, and Figure~\ref{fig:fixed_rates_vs_general}.

\setcounter{figure}{0}
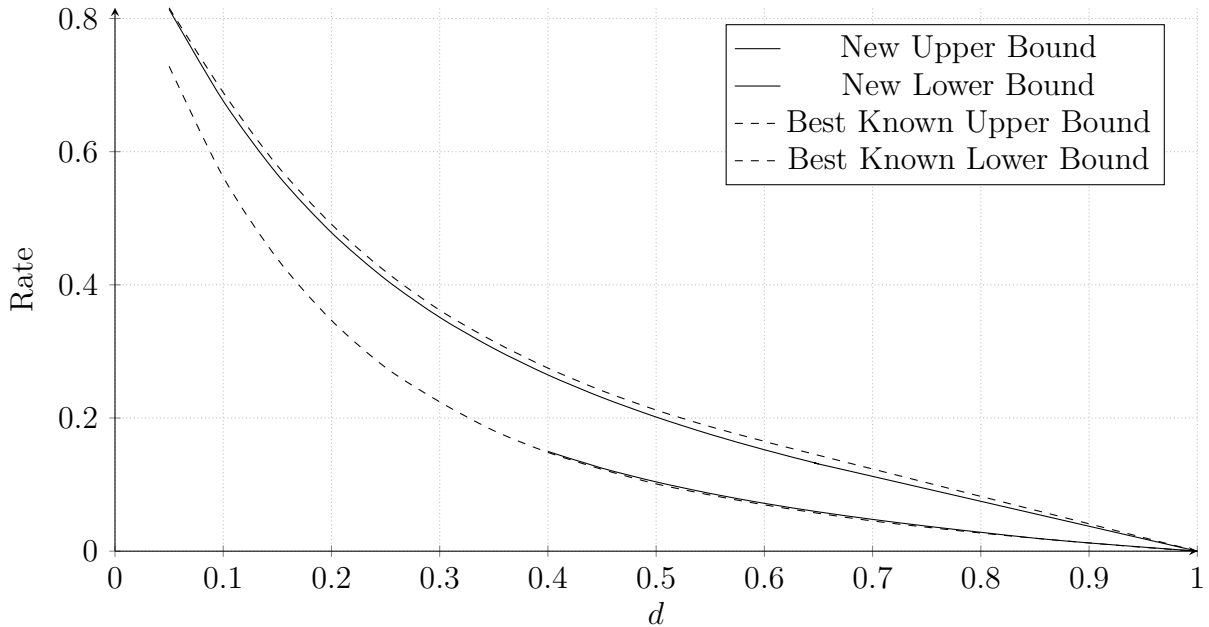
\begin{figure*}
	\centering
	\begin{tikzpicture}
	
	\begin{axis}[
	grid=both,
	width=450,
	height=250,
	tick label style={/pgf/number format/fixed },
	axis lines = left,
	xlabel = $d$,
	ylabel = {Rate},
	xtick={0, 0.1, 0.2, 0.3, 0.4,0.5, 0.6,0.7, 0.8, 0.9 ,1},
	legend pos=north east,
	grid style={densely dotted},
	smooth
	]
	
	\addlegendentry{New Upper Bound}
	\addplot [color=black, style=solid] coordinates {
		(0.05,0.8138427610328378)(0.10,0.6761893916803339)(0.15,0.5659635686287017)(0.20,0.478504314784656)(0.25,0.4082802369406138)(0.30,0.3512551582347527)(0.35,0.3044149196294756)(0.40,0.26476612791560644)(0.45,0.23083180275416795)(0.50,0.20144392019579563)(0.55,0.1754553520959549)(0.60,0.15232375237041654)(0.65,0.1312215062616179)(0.68,0.11983751366976653)(1,0)
	};
	
	\addlegendentry{New Lower Bound}
	\addplot [color=black, style=solid] coordinates {
		(0.40,0.149810)
		(0.45,0.124700)
		(0.50,0.104075)
		(0.55,0.086712)
		(0.60,0.071838)
		(0.65,0.059012)
		(0.70,0.047726)
		(0.75,0.037593)
		(0.80,0.028371)
		(0.85,0.019531)
		(0.90,0.012379)
		(0.95,0.006109)
		(1.0,0.0)
	};
	
	\addlegendentry{Best Known Upper Bound}
	\addplot [style=dashed]  coordinates {
		(0.05,0.816)(0.10,0.689)(0.15,0.579)(0.20,0.491)(0.25,0.420)(0.30,0.362)(0.35,0.315)(0.40,0.275)(0.45,0.241)(0.50,0.212)(0.55,0.187)(0.60,0.165)(0.65,0.144)(1,0)
	};

	\addlegendentry{Best Known Lower Bound}
	\addplot [style=dashed] coordinates {
		(0.05,0.72829)(0.1,0.56196)(0.15,0.43918)(0.2,0.34669)(0.25,0.27588)(0.3,0.2243)(0.35,0.18101)(0.4,0.14841)(0.45,0.12286)(0.5,0.101086)(0.55,0.084323)(0.6,0.069564)(0.65,0.056858)(0.7,0.045324)(0.75,0.035984)(0.8,0.027266)(0.85,0.019380)(0.9,0.012378)(0.95,0.005741)(1,0)
	};
	
	\addplot [] coordinates{(0,0)};
	\end{axis}
	\end{tikzpicture}
	\caption{Our new upper and lower bounds compared to previous state-of-the-art upper and lower bounds.
		The best known lower bound is given in~\cite{drinea2007improved,mitzenmacher2006simple} and the best known upper bounds are taken from~\cite{diggavi2007capacity,fertonani2010novel,dalai2011new,rahmati2014upper}} \label{fig:fixed_rates_vs_general}
\end{figure*}

\subsection{Organization}\label{subsec:organization}
In Section~\ref{sec:upper-bound}, we describe the main components of the Blahut-Arimoto algorithm and the best known upper bounds.
In Section~\ref{sec:efficient-baa}, we construct a memory efficient version of the Blahut-Arimoto algorithm and show that it can be used to prove tighter upper bounds on the capacity of the BDC.
Finally, in Section~\ref{sec:lower-bound} we present a heuristic approach to optimizing the input to Mitzenmacher and Drinea's construction, allowing us to prove tighter lower bounds on the capacity of the BDC.

In the interest of reproducibility, we have made our code publicly available through github.
The code used to generate the upper bounds in Sections~\ref{sec:upper-bound} and~\ref{sec:efficient-baa} is available~\href{https://github.com/ittai-rubinstein/BDC_Upper_Bounds}{here}.
The code used to generate our lower bounds (see Section~\ref{sec:lower-bound}) is available~\href{https://github.com/ittai-rubinstein/BDC_Lower_Bounds}{here}.

	\section{Upper Bound Theory}\label{sec:upper-bound}
	In this section, we will give an overview of the theory involved in our upper bound analysis.
In Section~\ref{subsec:discrete-memoryless-channels-and-the-baa}, we will introduce discrete memoryless channels and the Blahut-Arimoto algorithm which can be used to compute their capacity.
Then, in Section~\ref{subsec:auxiliary-deletion-channels}, we will introduce several auxiliary channels which will aid our analysis of the BDC.
We will show that these auxiliary channels are DMCs, allowing us to use the Blahut-Arimoto algorithm to find their capacity, and that their capacities are related to those of the BDC.

\subsection{Discrete Memoryless Channels and the Blahut-Arimoto Algorithm}\label{subsec:discrete-memoryless-channels-and-the-baa}

Our main results are achieved through improvements and extensions of the Blahut-Arimoto algorithm for finding the capacity of discrete memoryless channels.
In this section, we will give a brief overview of these channels and the Blahut-Arimoto algorithm (for a more detailed review, see~\cite[Chapter 9]{yeung2008information}).

A {\em discrete memoryless channel} (DMC) is defined by a finite input alphabet $\II$, a finite output alphabet $\OO$, and a transition probability matrix $P \in \Rplus^{\abs{\II}\times \abs{\OO}}$.
The channel sends each input letter $\ii$ to each output letter $\oo$ with probability $P_{\ii\rightarrow\oo}$ independently at random.
We note that the BDC is not a memoryless channel, so the results discussed here cannot be used to bound its capacity directly.
In Section~\ref{subsec:auxiliary-deletion-channels}, we will show that there exists a family of DMCs whose capacities converge on the capacity of the BDC from above.

In general, the capacity $\cC$ of a channel represents the maximum rate at which information can be communicated reliably over it~\cite{yeung2008information}.
For DMCs, the capacity is given by 
\begin{align}
\begin{split}
\label{eq:DMC-capacity}
\cC = \max_{\stackrel{\PP}{\ii \leftarrow \PP}} \left\{I(\ii; \oo)\right\} = \sup_{\PP \in \Rplus^{\II}} \left\{\sum_{\ii, \oo} \PP_i P_{\ii\rightarrow \oo} \log \left( \frac{P_{\ii\rightarrow \oo}}{\sum_{\ii^\prime} \PP_{\ii^\prime} P_{\ii^\prime \rightarrow \oo}} \right)\right\}
\end{split}
\end{align}
 where the supremum is taken over input distributions $\PP(\ii)$ that have a nonzero probability to output any letter $\ii\in\II$.
Both the input and the output alphabets are of finite size, so the information rate associated with any given input distribution $\PP$ can be computed using the right-hand-side of eq.~\eqref{eq:DMC-capacity}.
However, it is not obvious that one can efficiently find a distribution $\PP$ that maximizes this formula.

The {\em Blahut-Arimoto algorithm} (BAA - see Algorithm~\ref{alg:BAA-alg}) is an iterative algorithm which rapidly converges to the capacity for any given DMC.
It was introduced independently by Blahut and Arimoto in~\cite{blahut1972computation,arimoto1972algorithm}, who showed that it can be used to quickly find an input distribution whose rate is arbitrarily close to optimal.

\begin{figure}
	\begin{algorithm}[H] \label{alg:BAA-alg}
		\SetAlgoLined
		\DontPrintSemicolon
		\LinesNumberedHidden
		
		\SetKwInOut{Input}{input}
		\SetKwInOut{Output}{output}
		\SetKwRepeat{Do}{do}{while}
		
		\SetNlSty{large}{[}{]}
		\LinesNumberedHidden
		\Input{ Input and output alphabets $\II, \OO$\;
			$P:\left[0,1\right]^{\II\times \OO}$ - the channel's transition probability matrix \;
			$a$ - capacity approximation parameter}
		\Output{$\RR\in\R$ - an upper bound on the capacity of this channel\;
			$\PP \in \R^\II$ - an input distribution with a nearly optimal rate for this channel\;
		}
		\nlset{1} Set $\ell = 0$ and set $\PP^{(1)}$ to be the uniform distribution. \;
		\nlset{2} \Do{$\max_\ii \log_2 \left( \frac{\PP^{(\ell+1)}(\ii)}{\PP^{(\ell)}(\ii)} \right) \geq a$}
		{
			\begin{enumerate}
				\item $\ell \leftarrow \ell+1$
				\item For every $\ii \in \II$ and $\oo \in \OO$ compute
				\[
				\QQ^{(\ell)}_{\oo, \ii} = \frac{\PP^{(\ell)} (\ii) P_{\ii\rightarrow \oo}}{\sum_{\ii^\prime}\PP^{(\ell)}(\ii^\prime) P_{\ii^\prime \rightarrow \oo}}
				\]
				\item For every $\ii\in \II$, compute
				\[
				\PP^{(\ell + 1)}(\ii) = \frac{\prod_\oo (\QQ_{\oo,\ii}^{(\ell)})^{P_{\ii\rightarrow \oo}}}{\sum_{\ii^\prime} \prod_y (\QQ_{\oo, \ii^\prime}^{(\ell)})^{P_{\ii^\prime \rightarrow \oo}}}
				\]
				where the product is over all $\oo$ such that $\QQ_{\oo, \ii}^{(\ell)} > 0$ and $P_{\ii\rightarrow \oo} > 0$.
			\end{enumerate}
			
		}
		\nlset{3} Compute
		\[
		\RR = \sum_\ii \sum_\oo \PP^{(\ell)}(\ii)\cdot P_{\ii \rightarrow \oo} \cdot \log \left( \frac{\QQ^{(\ell)}_{\oo, \ii}}{\PP^{(\ell)}(\ii)}\right)
		\]
		where the sums are over all $\ii, \oo$ such that $\QQ_{\oo, \ii}^{(\ell)} > 0$ and $P_{\ii \rightarrow \oo} > 0$.

		\nlset{4} Return $\PP^{(\ell + 1)}(\ii), \RR$.
		\caption{The Blahut-Arimoto algorithm.\\
		}
	\end{algorithm}
	\captionsetup{labelformat=empty}
	\caption{}
\end{figure}

We denote by $\PP^{(\ell)}$ the input distribution used in the $\ell$th iteration, and by $\QQ^{(\ell)}_{\oo, \ii}$ the transition probabilities from the output alphabet to the input alphabet which can be thought of as the transitions of a probabilistic decoding algorithm.
The algorithm runs by performing an alternating maximization on $\PP$ and $\QQ$.

\begin{claim} \cite[Corollary 1]{arimoto1972algorithm}
	Let $a>0$. The BAA algorirthm converges within $\mathcal{O}(1/a)$ iterations, and returns a distribution $\PP$ with an information rate of at least $\RR \geq \cC(\textup{Ch})-a$.
	
\end{claim}

\subsection{Auxiliary Deletion Channels}\label{subsec:auxiliary-deletion-channels}
In Section~\ref{subsec:discrete-memoryless-channels-and-the-baa}, we defined discrete memoryless channels and presented the Blahut-Arimoto algorithm which can be used to bound their capacities.
In this section, we will present several auxiliary channels.
Originally introduced by Fertonani and Duman~\cite{fertonani2010novel} and also used by Dalai~\cite{dalai2011new}, these auxiliary channels are both discrete memoryless channels, allowing us to bound their capacities with the Blahut-Arimoto algorithm, and related to the BDC, allowing us to use them to bound the capacity of the BDC.

Let $W_n^d$ be the binary deletion channel with fixed input length $n$.
Denote by $\cC_{n}(d) \coloneqq \max I(X_1^n; Y)$ its capacity where the maximum is taken over all input distributions on $n$ bits, and $Y = W_n^d(X_1^n)$ (i.e., the output of the channel on input $X$).
\begin{claim}~\cite[Lemma 1]{dalai2011new}
	For any $n\in\N$ and $d\in [0,1]$,
	\[
	\cC (\text{BDC}_d) \leq \frac{1}{n} \cdot \cC_{n}(d) \;.
	\]
\end{claim}

Similarly, we denote by $W_{n,k}$ the channel with $n$-bit input whose output is uniformly chosen within the $\binom{n}{k}$ $k$-bit subsequences of the input (i.e. the channel deletes $n-k$ of the bits at random).
Denote its capacity by $\cC_{n,k} := \max I(X_1^n; Y)$ where $Y=W_{n,k}(X_1^n)$.

Fertonani and Duman, and then Dalai showed that the values of $\cC_{n,k}$ can be used to bound $\cC_{n}(d)$ from above with the following inequality.
	\begin{equation}
	\cC_n(d) \leq \sum_{k=1}^n \binom{n}{k} d^{n-k}(1 - d)^{k} \cdot \cC_{n,k} \label{eq:c-n-d-bound}
	\end{equation}

	\begin{proof}[Proof of \eqref{eq:c-n-d-bound}]
	The proof of this inequality is contained in the proof Lemma 4 in~\cite{dalai2011new}.
	We repeat the details for completeness. Let $X_1^n$ be an optimal distribution for $W_n^d$ and let $Y = W_n^d(X_1^n)$ be the output distribution of the channel on $X_1^n$. Define the random variable $L = |Y|$ (the length of $Y$).
	Let $S_d(n,k)$ be the probability that a string of length $n$ is transformed into a string of length $k$ after being transmitted through $W_n^d$, namely,
	\[
	S_d(n,k) = \binom{n}{k} d^{n-k}(1 - d)^{k} \;.
	\]
	
	Then,
	\begin{equation*}
		\begin{aligned}
			\cC_n(d) =  I(X_1^n;Y|L) = \sum_{k=0}^n S_d(n,k) \cdot I(X_1^n;Y|L = k)  \leq \sum_{k=1}^n S_d(n,k) \cdot \cC_{n,k}
		\end{aligned}
	\end{equation*}
	Where the first equality is due to the fact that $X \rightarrow Y \rightarrow L$ is a Markov chain. The inequality follows since $I(X_1^n;Y|L = k) \leq \cC_{n,k}$ and $\cC_{n,0} = 0$. 
\end{proof}

Clearly, $\cC_{n,k}$ describes the capacity of a channel with a finite input and a finite output size.
Therefore, we can describe it with a transition matrix $P \in \Rplus^{2^n \times 2^k}$ and use the BAA algorithm to compute its capacity to within a given additive error. In \cite{fertonani2010novel}, Fertonani and Duman, bounded $\cC_{n,k}$ for all $1 \leq k\leq n\leq 17$, and for the specific case where $k=n-1$ up to $n=22$.
However, they did not manage to go beyond these parameters due to the high computational complexity required~\cite{fertonani2010novel}.
In particular, we note the high memory complexity of the BAA which is the main bottleneck in applying it to channels with large alphabets.

In Section~\ref{sec:efficient-baa}, we will discuss several ways of decreasing the space complexity of the na\"ive BAA implementation (Algorithm~\ref{alg:BAA-alg}). Before that, we present several properties of $\cC_{n,k}$ that we will use when computing the upper bound on $\cC_n(d)$.

\begin{lemma} \cite[Lemma 1]{fertonani2010novel} \label{lem:C-n-k-first-bound}
	For all $n$ and $k$, it holds that $\cC_{n + 1,k}  \leq \cC_{n, k}$.
\end{lemma}
\begin{lemma} \cite[Lemma 3]{fertonani2010novel} \label{lem:C-n-k-second-bound}
	For all $n$ and $k\geq 1$, it holds that
	\[
	\cC_{n + 1,k} \leq \cC_{n,k-1} \cdot \left( 1 - \frac{k}{n+1} \right) + \left( \cC_{n,k} + 1 \right) \cdot \frac{k}{n+1} \;.
	\]
\end{lemma}
The following claim can be seen as a straight forward generalization of \autoref{lem:C-n-k-second-bound}.
\begin{lemma}
	\label{lem:C-n-k-third-bound}
	For every $s\in [n]$, it holds that
	\[
	\cC_{n,k} \leq \sum_{i=0}^s \frac{\binom{s}{i} \cdot \binom{n-s}{k -i}}{\binom{n}{k}} \cdot \left( \cC_{s,i} + \cC_{n-s, k-i}\right)
	\]
\end{lemma}
\begin{proof}
	
	Let $X_1^n$ be an optimal distribution for the channel $W_{n,k}$.
	Denote by $R_s$ the random variable that corresponds to the number of bits that survived the transmission of the last $s$ bits of $X$.
	It holds that
	\begin{equation*}
		\begin{aligned}
			\cC_{n,k} &= I(X_1^n;Y) \leq I(X_1^n;Y|R_s)  = \sum_{i=0}^s \Pr[R_s = i] \cdot  I(X_1^n;Y|R_s=i)
		\end{aligned}
	\end{equation*}
	where the inequality follows since $R_s$ and $X_1^n$ are independent.
	Now, $\Pr[R_s = i] = \frac{\binom{s}{i} \cdot \binom{n-s}{k -i}}{\binom{n}{k}}$ so we are left to show that $I(X_1^n;Y|R_s=i) \leq \cC_{s,i} + \cC_{n-s, k-i}$. Indeed, let $Z^{(0)}\coloneqq W_{n-s,k-i}(X_1^{n-s})$, $Z^{(1)} \coloneqq W_{s,i}(X_{n-s+1}^{n})$, and consider the following markov chain $X_1^n \rightarrow (Z^{(0)}, Z^{(1)}) \rightarrow Z$ where the second step is just the concatenation of $Z^{(0)}$ and $Z^{(1)}$.
	Thus, by the data processing inequality,
	\begin{equation*}
		I(X_1^n;Y|R_s=i) = I(X_1^n;Z) \leq I(X_1^n; (Z^{(0)}, Z^{(1)})) \leq \cC_{s,i} + \cC_{n-s, k-i} \;. \qedhere
	\end{equation*}
	
\end{proof}

	\section{Memory Efficient Implementation of the Blahut-Arimoto Algorithm}\label{sec:efficient-baa}
	Our ultimate goal is to run the Blahut-Arimoto algorithm on these auxiliary channels with larger input and output alphabets in order to find tighter upper bounds on the capacity of the BDC.

The main bottleneck in applying Algorithm~\ref{alg:BAA-alg} to larger input alphabets ($\II$) and output alphabets ($\OO$), is its memory requirement.
Recall that Algorithm~\ref{alg:BAA-alg} stores two  $\abs{\II\times\OO}$-sized arrays:
the transition probability matrix of the channel and the intermediate array $\QQ^{(\ell)}_{\oo, \ii}$ computed in each iteration.
This requires us to store and frequently access large arrays of floating point numbers, which increases the memory complexity of the algorithm, and poses the main technical limitation on our ability to scale to larger alphabets.

To get a sense of how quickly the memory usage of Algorithm~\ref{alg:BAA-alg} becomes impractical, assume that we want to compute upper bound on $\cC_{25,12}$ using the BAA algorithm (Algorithm~\ref{alg:BAA-alg}).
In this case, $\abs{\II}=2^{25}, \abs{\OO}=2^{12}$ and saving just the transition matrix, $P$, where each entry is represented by a {\em double} type floating point, would require a terabyte of RAM storage!

In the following section, Section~\ref{subsec:parallel-baa}, we use a combination of caching computation results and loop nest optimization to significantly reduce the asymptotic memory requirements of the BAA algorithm, albeit with a small increase in runtime. 
In Section~\ref{sec:using-sparsity}, we present a ``sparse'' version of Algorithm~\ref{alg:BAA-alg} that enables more efficient storage of the arrays. 
By employing these techniques, we were able to compute upper bounds on $\cC_{n,k}$ for values of $n$ and $k$ that were unknown prior to this work. Finally, in Section~\ref{subsec:results}, we present our computation which takes as input the upper bounds on $\cC_{n,k}$ and produces the improved upper bounds on the capacity of the BDC.

\subsection{Caching and Loop Nest Optimization}
\label{subsec:parallel-baa}

The main method we use to compute the bulk of our $\cC_{n,k}$ capacity bounds is by removing the requirement for storing the $\abs{\II\times\OO}$ sized matrices $P_{\ii\rightarrow\oo}, \QQ_{\oo, \ii}$ altogether, and replacing them with asymptotically smaller arrays instead.

We can remove the memory requirements for the transition probability matrix from the Blahut-Arimoto algorithm by accepting an oracle access to its entries instead.
Doing this requires a method of computing the channel's transition probabilities quickly.
The transition probabilities are proportional to the number of ways in which a given output string can be represented as a subsequence of a given input string, and the common approach to this computation would be to use the dynamic programming method shown in Algorithm~\ref{alg:dynamic-programming-transitions}.

\begin{figure}
	\begin{algorithm}[H]
		\label{alg:dynamic-programming-transitions}
		\SetAlgoLined
		\DontPrintSemicolon
		\LinesNumberedHidden
		
		\SetKwInOut{Input}{input}
		\SetKwInOut{Output}{output}
		\SetKwRepeat{Do}{do}{while}
		
		\SetNlSty{large}{[}{]}
		\LinesNumberedHidden
		\Input{ Input string $\ii\in \II \in \{0,1\}^n$\;
			Output string $\oo\in \OO \in \{0,1\}^k$}
		\Output{Transition probability $P_{\ii\rightarrow\oo}$}
		
		\nlset{1} Initialize a 2D array $A$ of size $(n+1)\times (k+1)$ with all values as 0\;
		\For{$i$ from $0$ to $n$}{
			\For{$j$ from $0$ to $k$}{
				\If{$i = 0$ and $j = 0$}{
					Set $A[i][j] := 1$\;
				}
				\If{$i > 0$ and $j = 0$}{
					Set $A[i][j] := A[i-1][j]$\;
				}
				\If{$i > 0$ and $j > 0$ and $\ii_i = \oo_j$}{
					Set $A[i][j] := A[i-1][j-1] + A[i-1][j]$\;
				}
				\If{$i > 0$ and $j > 0$ and $\ii_i \neq \oo_j$}{
					Set $A[i][j] := A[i-1][j]$\;
				}
			}
		}
		\nlset{2} Return
		\[
		\frac{A[n][k]}{\binom{n}{k}}
		\]
		\caption{A dynamic programming algorithm for computing the transition probabilities of the $\cC_{n,k}$ channel.\\
			This algorithm has time and memory complexity $\Theta\left(nk\right)$.}
	\end{algorithm}
	\captionsetup{labelformat=empty}
	\caption{}
\end{figure}

However, Algorithm~\ref{alg:dynamic-programming-transitions} runs in time $\Theta\left(n\cdot k\right)$, which proved to be too slow in practice.
Therefore, we used a slightly different approach shown in Algorithm~\ref{alg:cache-based-transitions} which runs in time $O(k)$, but requires a precomputed cache table of size $\Theta\left(2^{k+\ceil{n/2}}\right)$ that are computed using Algorithm~\ref{alg:dynamic-programming-transitions}.

Using large caches might seem counter-intuitive when our goal is to reduce the memory complexity of an algorithm, but this added memory is far smaller than the original memory complexity.
For instance, in the numerical example given above with $n=25$ and $k=12$, this precomputed table would take up only about 256MB of RAM, making it much more practical.

\begin{figure}
	\begin{algorithm}[H]
		\label{alg:cache-based-transitions}
		\SetAlgoLined
		\DontPrintSemicolon
		\LinesNumberedHidden
		
		\SetKwInOut{Input}{input}
		\SetKwInOut{Output}{output}
		\SetKwRepeat{Do}{do}{while}
		
		\SetNlSty{large}{[}{]}
		\LinesNumberedHidden
		\Input{ Input string $\ii\in \II = \{0,1\}^n$\;
			Output string $\oo\in \OO = \{0,1\}^k$\;
			Cache table $P_{\ii^\prime \rightarrow \oo^\prime}$ of all transition probabilities of $\cC_{\leq \ceil{n/2}, \leq k}$ channels}
		\Output{Transition probability $P_{\ii\rightarrow\oo}$}
		\nlset{1} Set $n_1 = \ceil{\frac{n}{2}}$ and $n_2 = \floor{\frac{n}{2}}$.\;
		\nlset{2} Define $\ii_1 = \ii_{:n_1}$ and $\ii_2 = \ii_{n_1+1:}$ to be the first $n_1$ and the last $n_2$ bits in $\ii$ (resp). \;
		\nlset{3} Return
		\[
		P_{\ii\rightarrow\oo} = \sum_{0 \leq k^\prime \leq k} P_{\ii_1 \rightarrow \oo_{:k^\prime}} P_{\ii_2 \rightarrow \oo_{k^\prime+1:}}
		\]
		\caption{An algorithm for computing the transition probabilities of the $\cC_{n,k}$ channel.
		}
	\end{algorithm}
	\captionsetup{labelformat=empty}
	\caption{}
\end{figure}

To avoid storing the intermediate array $\QQ^{(\ell)}_{\oo, \ii}$, we use an algorithmic concept called loop nest optimization~\cite{wolf1992improving}, meaning that we alter the order in which the operations of the algorithm are performed, in order to reduce its memory storage and retrieval requirements.
We perform this optimization to construct Algorithm~\ref{alg:LNO-BAA-alg}, where we avoid computing the entries of $\QQ^{(\ell)}_{\oo, \ii}$.
Instead, we compute the same outputs of each iteration of the Blahut-Arimoto algorithm using only the $\mathcal{O}\left(\abs{\II}+\abs{\OO}\right)$ sized intermediate arrays $\BAADenominator^{(\ell)}_{\oo}, \BAAWeight^{(\ell)}_{\ii}$.

Since each step of the Blahut-Arimoto optimization is logically unchanged, Algorithm~\ref{alg:LNO-BAA-alg} converges with exactly the same number of iterations as Algorithm~\ref{alg:BAA-alg}.
In both algorithms, each iteration has $\Theta\left(\abs{\II\times\OO}\right)$ arithmetic operations and Algorithm~\ref{alg:LNO-BAA-alg} also requires $\Theta\left(\abs{\II\times\OO}\right)$ queries to the transition probability oracle, each of which takes $\mathcal{O}(k)$ time (when implemented using Algorithm~\ref{alg:cache-based-transitions}).
Therefore, the total runtime of our optimized BAA Algorithm~\ref{alg:LNO-BAA-alg} is $\mathcal{O}(k) \cdot\Theta\left(\abs{\II\times\OO}\right) \cdot \mathcal{O} (1/a) =  \mathcal{O} \left(k \cdot \frac{2^{n+k}}{a}\right)$.

Using Algorithm~\ref{alg:LNO-BAA-alg}, we computed bounds on $\cC_{n,k}$ for all $n + k\leq 39$ where $n\in [28]$.

\begin{figure}
	\begin{algorithm}[H]
		\label{alg:LNO-BAA-alg}
		\SetAlgoLined
		\DontPrintSemicolon
		\LinesNumberedHidden
		
		\SetKwInOut{Input}{input}
		\SetKwInOut{Output}{output}
		\SetKwRepeat{Do}{do}{while}
		
		\SetNlSty{large}{[}{]}
		\LinesNumberedHidden
		\Input{ Input and output alphabets $\II, \OO$\;
			$P:\II\times \OO \rightarrow [0,1]$ - an oracle to the channel transition probability \;
			$a$ - capacity approximation parameter}
		\Output{$\PP \in \R^\II$ - an input distribution with a nearly optimal rate for this channel\;
			$\RR\in\R$ - the information rate of $\PP$\;
		}
		\nlset{1} Set $\ell = 0$ and set $\PP^{(1)}$ to be the uniform distribution. \;
		\nlset{2} \Do{$\max_\ii \log_2 \left( \frac{\PP^{(\ell+1)}(\ii)}{\PP^{(\ell)}(\ii)} \right) \geq a$}
		{
			\begin{enumerate}
				\item $\ell \leftarrow \ell+1$
				\item For every $\oo \in \OO$ compute
				\[
				\BAADenominator^{(\ell)}_{\oo} = \sum_{\ii \in \II}\PP^{(\ell)}(\ii) P_{\ii \rightarrow \oo}
				\]
				\item For every $\ii \in \II$ compute
				\[
				\BAAWeight^{(\ell)}_{\ii} = \prod_\oo \left(\frac{\PP^{(\ell)} (\ii) P_{\ii\rightarrow \oo}}{\BAADenominator^{(\ell)}_{\oo}}\right)^{P_{\ii\rightarrow \oo}}
				\]
				\item For every $\ii\in \II$, compute
				\[
				\PP^{(\ell + 1)}(\ii) = \frac{\BAAWeight^{(\ell)}_{\ii}}{\sum_{\ii^\prime} \BAAWeight^{(\ell)}_{\ii^\prime}}
				\]
			\end{enumerate}
			
		}
		\nlset{3} Compute
		\[
		\RR = \sum_\ii \sum_\oo \PP(\ii)\cdot P_{\ii \rightarrow \oo} \cdot \log \left( \frac{P_{\ii \rightarrow \oo}}{\BAADenominator^{(\ell)}_{\oo}}\right)
		\]
		Return $\PP^{(\ell + 1)}(\ii), \RR$.
		\caption{A loop nest optimized implementation of the Blahut-Arimoto algorithm.\\
		Note that this algorithm contains arithmetic operations over real numbers, which could lead to undefined intermediate values (e.g., due to a division by $0$).
		We work with the convention that the outputs of such operations are fixed to $0$.
		}
	\end{algorithm}
	\captionsetup{labelformat=empty}
	\caption{}
\end{figure}

\subsection{Using sparsity}
\label{sec:using-sparsity}
In this section, we will use the fact that when $k$ is close to $n$, the matrices $P$ and $\QQ^{(\ell)}_{\oo,\ii}$ are (very) sparse and thus can be represented in a compressed form.
Recall that a major bottleneck in Algorithm~\ref{alg:BAA-alg} is the memory usage that it requires.
Consider for example the case where $k = n-1$.
In this case, the size of $P$ is $2^{2n-1}$ but each row of $P$ has at most $n$ nonzero values which implies that each row contains at least $2^{n-1} - n$ zeros. Also, recall that Algorithm~\ref{alg:BAA-alg} does not consider zero values in its computations.
This motivates us to try to modify Algorithm~\ref{alg:BAA-alg} and create a version that saves only the nonzero entries.

There are many ways of representing sparse matrices.
We shall use the compressed sparse column format (CSC), in which the sparse matrix is saved using three (one dimensional) arrays: the \emph{data} array, the \emph{col-ind} array and the \emph{row-ind} array.
The data array will store all the nonzero values of the sparse matrix and the row-ind will store the corresponding row indices of these values, respectively.
The col-ind at index $j$ will encode the total number of nonzero values before the $j$th column (the size of col-ind is the number of columns in the matrix plus $1$).
Another representation that we will use is the compressed sparse row format (CSR), that is similar to the CSC format.
Here, the col-ind and row-ind array change roles.
The col-ind will store the column indices of the data and  the row ind at index $j$ will encode the total number of nonzero values above the $j$th row.

In our implementation, we extend the CSC representation by incorporating two additional arrays to compress the row data in a manner similar to the CSR  format. 
The first of these arrays, called \emph{perm-data}, stores the permutation that maps the data array of the CSC matrix to the data array of a CSR matrix. 
The second array, called \emph{row-ind-CSR}, represents the row indices of the corresponding CSR matrix.

The modifications to Algorithm~\ref{alg:BAA-alg} are:
\begin{itemize}
	\item The input to the algorithm is now the channel transition matrix $P$ in the CSC format. 
	\item Prior to the loop in Step 2, we compute the two additional arrays discussed above: the perm-data array and row-ind-CSR array.	
	\item In Step 2, we loop over only the nonzero values of $P$ and $\QQ^{(\ell)}{\oo,\ii}$ to compute $\QQ^{(\ell)}{\oo,\ii}$ and $\PP^{(\ell+1)}(x)$. The extension arrays we added allow us to iterate over the rows and columns of the sparse matrix.
	
\end{itemize}

Another major advantage of the modifications is the significantly improved runtime of the algorithm. 
By looping over just the non-zero values of $\QQ^{(\ell)}_{\oo,\ii}$ and $\PP^{(\ell+1)}$, the number of arithmetic operations becomes linear in the sparsity order, rather than linear in the size of the matrices. 
Specifically, the runtime of the sparse BAA algorithm is $\mathcal{O}(1/a) \cdot \mathcal{O} \left( \binom{n}{k} \cdot 2^n\right) = \mathcal{O} \left(\frac{n^k \cdot 2^n}{a}\right)$.

For example, when $k = n-2$, computing $\QQ^{(\ell)}_{\oo,\ii}$ and $\PP^{(\ell)}$ in Step 2 takes $O(n^2 \cdot 2^n)$ time in the sparse implementation, while the na\"ive implementation would take $O(2^{2n - 2})$ time
Using these modifications, we were able to compute $\cC_{n,k}$ for all $k,n$ such that $n \in [18, 24]$ and $k\in[n-3, n-1]$.

\subsection{Results}\label{subsec:results}

In \Cref{sec:upper-bound}, we showed how the values $\cC_{n,k}$ for all $k\in [n]$ can be used to obtain an upper bound on $\cC_n(d)$ via Equation~\eqref{eq:c-n-d-bound}.
By employing our improved BAA algorithms, we were able to compute many previously unknown $\cC_{n,k}$ values. 
Nevertheless, for $22\leq n \leq 28$, there are $\cC_{n,k}$ values that we could not compute even with the improved algorithms. 
To obtain upper bounds on these unknown $\cC_{n,k}$ values, we used the minimum value resulting from \autoref{lem:C-n-k-first-bound} and \autoref{lem:C-n-k-third-bound}.
Then, for every fixed $d$, we simply compute
\begin{equation} \label{eq:upper-bound-comp}
\cC(\textup{BDC}_{d}) \leq \min_{n\in [28]} \left( \frac{1}{n} \sum_{k=1}^{n} \binom{n}{k} d^{n-k}(1 - d)^{k} \cC_{n, k} \right)\;.
\end{equation}
To prove \autoref{thm:upper-bound-thm}, we will use Rahmati and Duman's result \cite[Theorem 1]{rahmati2014upper}.
\begin{theorem}\cite[Theorem 1]{rahmati2014upper} \label{thm:Rah-Dum-conv}
	Let $\lambda,d' \in [0,1]$ and denote $d = \lambda d' + 1 - \lambda$. Then, 
	\[
	\frac{\cC(\textup{BDC}_d)}{1 - d} \leq \frac{\cC(\textup{BDC}_d')}{1 - d'} \;.
	\]
\end{theorem}
When $d'=\delForUB$, we get that $\cC(\textup{BDC}_{d'})/(1 - d') = \ourBDCUpperBound$ and thus, for every $d \geq \delForUB$, we have that $\cC(\textup{BDC}_d) \leq \ourBDCUpperBound \cdot (1-d)$. Results for smaller values of $d$ are listed in \Cref{tab:upper-bound} and plotted in \Cref{fig:fixed_rates_vs_general}.

	\begin{table}
	\begin{center}
		\begin{tabular}{||c | c | c ||}
			\hline
			$p$ & \shortstack{New Upper Bound} & \shortstack{Best Known \\Upper Bound \cite{fertonani2010novel,diggavi2007capacity}} \\ [0.5ex]
			\hline
			$0.01$ & $ {\bf 0.9583} \;(n=24)$ & $0.963$ \\
			\hline
			$0.02$ & $ {\bf 0.9189} \;(n=24)$ & $0.926$ \\
			\hline
			$0.03$ & $ {\bf 0.8817} \;(n=24)$ & $0.891$ \\
			\hline
			$0.04$ & $ {\bf 0.8467} \;(n=24)$ & $0.858$ \\
			\hline
			$0.05$ & $ {\bf 0.8139} \;(n=24)$ & $0.816$ \\
			\hline
			$0.10$ & $ {\bf 0.6762} \;(n=22)$ & $0.689$ \\
			\hline
			$0.15$ & ${ \bf 0.5660 }\;(n=22)$ & $0.579$ \\
			\hline
			$0.20$ & $ {\bf 0.4786} \;(n=22)$ & $0.491$ \\
			\hline
			$0.25$ & $ {\bf 0.4083} \;(n=22)$ & $0.420$ \\
			\hline
			$0.30$ & $ {\bf 0.3513} \;(n=22)$ & $0.362$ \\
			\hline
			$0.35$ & $ {\bf 0.3045} \;(n=22)$ & $0.315$ \\
			\hline
			$0.40$ & $ {\bf 0.2648} \;(n=23)$ & $0.275$ \\
			\hline
			$0.45$ & $ {\bf 0.2309} \;(n=23)$ & $0.241$ \\
			\hline
			$0.50$ & $ {\bf 0.2015} \;(n=24)$ & $0.212$ \\
			\hline
			$0.55$ & $ {\bf 0.1755} \;(n=25)$ & $0.187$  \\
			\hline
			$0.60$ & $ {\bf 0.1524} \;(n=27)$ & $0.165$ \\
			\hline
			$0.65$ & $ {\bf 0.1313} \;(n=28)$ & $0.144$  \\
			\hline
			$0.68$ & $ {\bf 0.1199} \;(n=28)$ & $-$  \\
			[1ex]\hline
		\end{tabular}
	\end{center}
	\caption{Our upper bound compared to the upper bound computed in \cite{fertonani2010novel}. We note that for $d=0.05$ the previous best known upper bound is given in \cite{diggavi2007capacity}. Alongside the bound, we write the value of $n$ for which the minimum was obtained in left hand side of \eqref{eq:upper-bound-comp}.}
	\label{tab:upper-bound}
\end{table}

	\section{A new Lower Bound}\label{sec:lower-bound}
	
\subsection{The framework of Drinea and Mitzenmacher}\label{subsec:framework-of-drinea-and-mitzenmacher}
We give a brief description of the main framework given in~\cite{drinea2007improved}.
Their framework suits any channel that introduce i.i.d.\ deletions and also i.i.d.\ duplications.
Such a channel is described by a probability distribution $G$ over the nonnegative integers, where a nonnegative integer $j$ is sampled with probability $j$. 
The BDC$_d$ is defined with $G_0 = p$, $G_1 = 1 - p$, and for every $j\geq 2$, $G_j = 0$.
Another channel we consider in this work is the PRC$_{\lambda}$ that is defined as follows

\begin{definition}
	Let $\lambda > 0$.
	The \emph{Poisson repeat channel with parameter $\lambda$ (PRC$_{\lambda}$)} replaces each transmitted bit randomly (and independently of other transmitted bits), with a discrete number of copies of that bit, distributed according to the Poisson distribution with parameter $\lambda$.
\end{definition}

This channel was first defined by Mitzenmacher and Drinea in \cite{mitzenmacher2006simple} who used it to prove a lower bound of $\left(1-d\right)/9$ on the rate of  the BDC. 
What they observed is that a code for the PRC$_\lambda$ having rate $\mathcal R$, yields a code for the BDC$_d$ of rate $(1-d) \cdot \mathcal{R}/\lambda$.

Drinea and Mitzenmacher define their code using a {\em run length distribution} $\PP$, which samples each non-negative integer $j$ with probability $\PP_j$. \footnote{
	The distribution, $\PP$, must have a \emph{geometric decreasing tail}, that is, there are two real constants $c_{\PP}\in (0,1]$, $\alpha_{\PP} \in [0,1)$ and an integer constant $M_{\PP}$ such that (i) $\PP_j \leq c_{\PP}$ for $1\leq j\leq M_{\PP}$ and (ii)  $\PP_j \leq (1 - \alpha_{\PP}) \alpha_{\PP}^{j-1}$ for $j > M_{\PP}$.}

Each codeword of length $N$ in the code is constructed as follows.
The first symbol ($0$ or $1$) is chosen uniformly at random.
The rest of the codeword is generated by deciding on the lengths of the \emph{runs} that form it, where runs are defined as the maximal length substrings of the codeword of the same symbol (e.g., the string $110001$ consists of $3$ runs: $11$, $000$ and $1$, of lengths $2$, $3$, and $1$, respectively).
The lengths of the runs of the codeword are sampled i.i.d.\ from $\PP$ and the symbols of the runs are alternating.
The sampling stops when the length of the generating string is $\geq N$.
If the length of the generated string is strictly greater than $N$, it is truncated.

We now define the notion of \emph{types}.
Consider a transmitted codeword $X$ and let $Y$ be the output of the channel upon transmitting $X$. We write $X$ and $Y$ as a concatenation of their runs, that is, $X = r_0^X\circ r_1^X \circ \cdot \circ r_m^X$, $Y = r_0^Y\circ r_1^Y \circ \cdot \circ r_{m'}^Y$
\begin{definition}
	Let $r^{Y}_j$ be a run of length $k$.
	Let $j_1,\ldots, j_s$ be a sequence of consecutive indices such that,
	\begin{itemize}
		\item The first bit of $r^{Y}_j$ corresponds to the first bit of $r^X_{j_1}$ that was not deleted by the channel.
		\item The runs  $r^X_{j_2}, r^X_{j_4}, \ldots, r^X_{j_s-1}$ were deleted by the channel.
		\item The run $r^X_{j_{s+1}}$ is not completely deleted by the channel. 
	\end{itemize}
	Then, the \emph{type of $r^{Y}_j$} is an $s$ tuple that represent the lengths of the runs indexed at $j_1,\ldots, j_s$, namely
	$\left( \left|r^X_{j_1}\right|, \ldots, \left|r^X_{j_s}\right| \right)$.
\end{definition}

The probability of a type $t = (z, s_1, r_1, \ldots, s_i, r_i)$, is \cite[Equation 3]{drinea2007improved}
\[
\PP_z (1 - d^z) \left( \prod_{\ell=1}^{i} \PP_{s_{\ell}}\PP_{r_{\ell}} \right) d^{s} \;,
\]
where $s:=s_1+ \ldots s_i$ and $d:=G_0$.
Denote by $T$ and $K$ the random variables representing the length and type of a runs in $Y$, respectively.
It holds that \cite[Equation 34]{drinea2007improved},
\begin{equation} \label{eq:t-k-prc}
\Pr[T=t, K=k] = \PP_z \PP_{s_1}  \PP_{s_1} \cdots \PP_{s_i} \PP_{r_i} \cdot d^s \cdot \left( \rho_{z+r,k} - \rho_{r,k} \cdot d^z \right)
\end{equation}
where $\rho_{a,b}$ is the probability that $a\geq 1$ bits transmitted over a channel with distribution $G$ generate $b$ bits. Note that $\rho_{z,0} = G_0^z = d^z$. 

Let $Q_{n,m}$ be the probability that the length of $m$ consecutive runs is exactly $n$.
It holds that $Q_{n,m} = \sum_{\ell=1}^{n -m +1}\PP_{\ell}\cdot Q_{n-\ell, m - 1}$ where $Q_{0,0} = 1$.
Let $D = \sum_z \PP_z d^z$ be the probability that a run is deleted from $X$, the distribution $\mathcal{K}$ for run lengths in $Y$ is given by \cite[Equation 35]{drinea2007improved}

\begin{equation} \label{eq:run-dist-prc}
\mathcal{K}_k := \Pr[K = k] = \sum_{i=1}^{\infty} D^i \sum_{z=1}^{\infty} \sum_{r=i}^{\infty} \PP_z Q_{r,i} \cdot \left(\rho_{z+r,k} - \rho_{r,k} \cdot d^z \right) \;.
\end{equation}

Given the distribution $\PP$ we can compute the rate of the respective code constructed using this method with the following theorem.
\begin{theorem} \cite[Theorem 4]{drinea2007improved} \label{thm:poisson-cap-lbound}
	The capacity of the the channel defined by the distribution $G$ with $G_0 > 0$ is lower bounded by
	\begin{equation} \label{eq:rate-l-bound1}
	\frac{1}{\frac{1 + D}{1-D} \cdot \sum_{z}z \PP_z} \left[ \frac{1 + D}{1 - D}\cdot H(\PP) - \left( H(T, K) - H(\mathcal{K}) \right) \right]
	\end{equation}
	where $\Pr[T=t, K=k]$ is given in~\eqref{eq:t-k-prc} and $\mathcal{K}$ is given in \eqref{eq:run-dist-prc}.
\end{theorem}
Note that evaluating numerically~\eqref{eq:rate-l-bound1} for a given distribution $\PP$ is not an easy task since it involves infinite sums. 

In \cite{mitzenmacher2006simple,drinea2007improved}, the authors limit themselves to the case where $\PP$ is a geometric distribution. In this case, they managed to derive simpler expressions to ~\eqref{eq:rate-l-bound1} (see \cite[Theorem 2]{mitzenmacher2006simple} for the PRC$_{\lambda}$ and \cite[Corollary 1]{drinea2007improved} for the BDC$_d$). However, the authors do not provide an argument for why geometric distributions should be optimal.
Indeed, we will improve upon this lower bounds by constructing a better distribution $\PP$.

\subsection{Our heuristic approach to find better input distributions}\label{subsec:our-heuristic-approach}
In this section, we will present a heuristic approach to optimizing the input run length distribution $\PP$.
Note that even though our optimization of $\PP$ is heuristic, because $\PP$ is used only as parameters for Mitzenmacher and Drinea's rigorous construction, the resulting bounds are provably correct. Indeed, we plug in the resulting distribution in \Cref{thm:poisson-cap-lbound}.

Our approach is based on a heuristic score function $R_{\text{heur}}$, which should approximate the information rate of Mitzenmacher and Drinea's code for the given distribution.
We select the score function so that it can be optimized with a Blahut-Arimoto type algorithm.

Our main observation is that for the distributions used by Mitzenmacher and Drinea, only a small fraction of the runs are deleted.
If no runs were deleted, then we could view the channel as a DMC that maps input run lengths into output run lengths and its information rate would be given by eq.~\ref{eq:memoryless-rate}, where $P_{\ii\rightarrow\oo}$ is the transition probability matrix of input run lengths to output run lengths.
\begin{equation}
R_{\text{DMC}} = \sum_{\ii, \oo} \PP_i P_{\ii\rightarrow \oo} \log \left( \frac{P_{\ii\rightarrow \oo}}{\sum_{\ii^\prime} \PP_{\ii^\prime} P_{\ii^\prime \rightarrow \oo}} \right)\label{eq:memoryless-rate}
\end{equation}

Equation~\eqref{eq:memoryless-rate} overlooks two key effects: that the binary deletion channel can delete entire runs causing a loss of information, and the rate given in equation~\eqref{eq:memoryless-rate} is in bits {\em per run}.
To correct for the first effect, we subtract from $R_{\text{DMC}}$ a loss function $\Delta \cdot \E_{\ii \leftarrow \PP} D_\ii$ where $\Delta$ is a parameter and $D_\ii$ is the probability that a run of length $\ii$ will be deleted by the channel.
To correct for the latter effect, we normalize the rate by the inverse of the average run length $L(\PP) = \E_{\ii\leftarrow\PP} \ii$.

After correcting for these effects, we get our heuristic score function in eq.~\eqref{eq:information-rate-heuristic0}.
In Section~\ref{subsec:analysis-of-the-heuristic}, we give a more detailed analysis of these approximations.
\begin{align}
\begin{split}
\label{eq:information-rate-heuristic0}
R_{\text{heur}} &\defeq \frac{1}{\sum_r \PP_r r} \left[ R_{\text{DMC}} - \sum_{r} \PP_r \Delta D_r \right]
\end{split}
\end{align}

In Section~\ref{subsec:optimizing-the-heuristic-score-function}, we show how a Blahut-Arimoto type algorithm can be used to maximize eq.~\eqref{eq:information-rate-heuristic0} w.r.t the input distribution $\PP$ in polynomial time.
Roughly speaking, this optimization works by enumerating over polynomially many potential values of $L(\PP)$, and using an extension of the BAA similar to the one used in~\cite{li2019blahut} to optimize $\PP$ under the different constraints on the average run lengths.

In practice, this enumeration converges too slowly, so we use a heuristic basin-hopping optimization algorithm to select the external parameters $\Delta, L_0$.
For each setting of $\Delta, L_0$, we use a BAA type algorithm to find a distribution $\PP(\Delta, L_0)$ which maximizes eq.~\eqref{eq:information-rate-heuristic0} under the constraint of $L(\PP)=L_0$.
We use the basin-hopping optimization to select the parameters for which the distribution $\PP(\Delta, L_0)$ yields the best lower bound from Mitzenmacher and Drinea's construction.

Recall that we use this heuristic only to optimize the parameters needed for applying Mitzenmacher and Drinea's rigorous construction.
Therefore, the lower-bounds obtained by this method are rigorous (though not necessarily tight).

We use this heuristic approach to optimize the input distribution for different deletion probabilities $d\in (0,1)$, giving us an improved lower bound in the intermediate deletion probability regime $d\in[0.4, 0.9]$ (see  Table~\ref{tab:lower-bound}).
In order to lower bound the capacity of the BDC in the $d\rightarrow 1$ regime, we use the same technique to lower bound the capacity of the PRC with a given parameter $\lambda_0 = 0.19$, yielding the bound
\[
\frac{\cC(\textup{BDC}_d)}{1-d} \geq \frac{\cC(\textup{PRC}_{\lambda_0})}{\lambda_0} > \frac{0.0232}{0.19} > 0.1221 \;.
\]

\begin{table}
	\begin{center}
		\begin{tabular}{||c | c | c | c||}
			\hline
			$d$ & \shortstack{Best previous \\ lower bound} & \shortstack{Direct lower\\bound} & \shortstack{PRC based\\lower Bound} \\ [0.5ex]
			\hline
			0.40 & 0.148410 & \textbf{0.149810} & 0.073313\\
			\hline
			0.45 & 0.122860 & \textbf{0.124700} & 0.067204\\
			\hline
			0.50 & 0.101860 & \textbf{0.104075} & 0.061094\\
			\hline
			0.55 & 0.084323 & \textbf{0.086712} & 0.054985\\
			\hline
			0.60 & 0.069564 & \textbf{0.071838} & 0.048875\\
			\hline
			0.65 & 0.056858 & \textbf{0.059012} & 0.042766\\
			\hline
			0.70 & 0.045324 & \textbf{0.047726} & 0.036657\\
			\hline
			0.75 & 0.035984 & \textbf{0.037593} & 0.030547\\
			\hline
			0.80 & 0.027266 & \textbf{0.028371} & 0.024438\\
			\hline
			0.85 & 0.019380 & \textbf{0.019531} & 0.018328\\
			\hline
			0.90 & 0.012378 & \textbf{0.012379} & 0.012219\\
			\hline
			0.95 & 0.005741 & 0.005631 & \textbf{0.006105}\\
			[1ex]\hline
		\end{tabular}
	\end{center}
	\caption{Our lower bound compared to the lower bounds computed in~\cite{mitzenmacher2006simple,drinea2007improved}.}
	\label{tab:lower-bound}
\end{table}

\subsection{Analysis of the heuristic}\label{subsec:analysis-of-the-heuristic}
In Section~\ref{sec:lower-bound}, we recalled Mitzenmacher and Drinea's approach to converting distributions on lengths of runs into error correcting codes for the binary deletion channel and the Poisson repeat channel~\cite{drinea2007improved}.
We gave a very high-level heuristic formula for the information rate a Mitzenmacher and Drinea's type code would achieve for a given run length distribution and optimized the input run length distribution using this heuristic formula.
In this section, we will give a slightly more formal derivation of the heuristic formula from Section~\ref{sec:lower-bound}.

Let $\PP$ be a distribution on run lengths to be input into Mitzenmacher and Drinea's construction.
If none of the input runs were completely deleted, then the channel could effectively be seen as a discrete memory-less channel on runs, and the information rate of the code (per run) would be given by
\begin{equation}
R_{\text{memoryless}} = \sum_{\ii, \oo} \PP_i P_{\ii\rightarrow \oo} \log \left( \frac{P_{\ii\rightarrow \oo}}{\sum_{\ii^\prime} \PP_{\ii^\prime} P_{\ii^\prime \rightarrow \oo}} \right)\label{eq:BAA_rate}
\end{equation}

However, this formula for the rate overlooks two key aspects of the deletion channel and the code.
Namely, that runs can be deleted (i.e., the BDC and the PRC are not memoryless channels) and that different run lengths have a different cost for the rate of the code.

Consider the effects of runs being deleted by the channel.
Every time a run is deleted by the channel, this causes the preceding and following input runs to be merged and results in a loss of information as the output run can no longer be assigned to a single input run.

Let $d$ denote the probability that the channel will delete any single bit ($d$ is the deletion rate for the deletion channel and $d=e^{-\lambda}$ for the $\textup{PRC}_\lambda$ channel), and let $D = \E_{\ii \leftarrow \PP} d^\ii$ denote the probability that any single input run is deleted.

The first approximation that we will make is that $D$ is small.
This assumption is reasonable for most useful distributions.
For instance, for the distributions used in~\cite{mitzenmacher2006simple, drinea2007improved}, the run deletion probabilities were around $D \approx 4.4\%$.

If we completely neglected the deletion probability $D$, then we would return to the formula for a memoryless channel and our optimization will no longer have any pressure to prefer longer runs which are less likely to be deleted by the channel, effectively increasing $D$ and possibly invalidating our assumption.
Therefore, we need to take into account some effects of order $O(D)$, but we allow ourselves to neglect those of order $O\left(D^2\right)$.
In other words, we will take into account the effects of output runs resulting from merging at most $2$ input runs.

Consider an input run of length $r$.
The probability that this run will be deleted by the channel is $D_r = d^r$.
Denote by $r_{\textup{before}}, r_{\textup{after}}$ the lengths of the runs immediately before and after it, respectively.
If the run $r$ is deleted, then the latter two ($r_{\textup{before}}, r_{\textup{after}}$) will be merged by the channel.

The possibility of runs being deleted or merged can lead to two adverse effects.
First, the added uncertainty of determining which runs were merged with which, can increase the entropy required for the decoding.
Due to the high degree of difficulty in estimating this effect and some evidence that it is less significant (see e.g.~\cite{kirsch2009directly}), we neglect it.
The second effect, on which we will focus most of our efforts, is that some amount of information $\Delta_I\left(r_{\textup{before}}, r_{\textup{after}}\right)$ is lost because we are given the output of the channel only on the merged run of length $r_{\textup{before}} + r_{\textup{after}}$ (and not on the individual runs of lengths $r_{\text{before}}, r_{\text{after}}$).
The loss of information is quantified in eq.~\eqref{eq:information-lost}, by comparing the information rate of two separate runs of lengths $r_1, r_2$, with the rate of a single merged run of length $r_1 + r_2$.

\begin{equation}
    \label{eq:information-lost}
    \begin{aligned}
        \Delta_I \left(r_1, r_2 \right) \defeq &\sum_{o_1, o_2} P_{r_1 \rightarrow o_1} P_{r_2 \rightarrow o_2} \left(\log{\left( \frac{P_{r_1\rightarrow o_1}}{\sum_{r} \PP_{r} P_{r\rightarrow o_1}} \right)} + \log{\left( \frac{P_{r_2\rightarrow o_2}}{\sum_{r} \PP_{r} P_{r\rightarrow o_2}} \right)}\right)  -\\
        &- \sum_{o} P_{r_1 + r_2 \rightarrow o} \log{\left( \frac{P_{r_1 + r_2 \rightarrow o}}{\sum_{\rho_1, \rho_2} \PP_{\rho_1} \PP_{\rho_2} P_{\rho_1+\rho_2\rightarrow o}} \right)}
    \end{aligned}
\end{equation}

Note that this information loss depends only on the lengths of the runs being merged, and that these lengths are independent of the length of the run being deleted.
Let $\Delta_I = \E_{r_1, r_2 \leftarrow \PP} {\Delta_I (r_1, r_2)}$ denote the average information lost due to such a merger.
Denote by $X$ the distribution of input codewords, by $Y$ the output codeword and let $\mathcal{T}$ denote the random variable containing the division of input runs into types (i.e.\ which runs in $X$ were merged due to deletions by the channel).
Putting our approximations into an information theoretic language, we have:

\begin{align}
\begin{split}
\label{eq:information-rate-approximation2}
\textup{Information Rate} &=
I\left( X ; Y \right) = H\left( X \right) - H\left( X \mid Y \right)\\
&= H\left( X \right) - H\left( \mathcal{T}, X \mid Y \right) + H\left( \mathcal{T} \mid X, Y \right) \\
&= H\left( X \right) - H\left( X \mid Y, \mathcal{T} \right) - H\left(\mathcal{T}\mid Y\right) + H\left( \mathcal{T} \mid X, Y \right) \\
&= \frac{\abs{X}}{\E_{r\leftarrow \PP} r} \left[ R_{\text{memoryless} } - D \Delta_I + O\left(D^2\right) \right] - H\left(\mathcal{T}\mid Y\right)+ H\left( \mathcal{T} \mid X, Y \right)
\end{split}
\end{align}

Equation~\eqref{eq:information-rate-approximation2} gives us our separation of the exact information rate into an approximate formula (eq.~\eqref{eq:information-rate-heuristic}) and error terms.

\begin{align}
\begin{split}
\label{eq:information-rate-heuristic}
R_{\text{heuristic}} &= \frac{1}{\E_{r\leftarrow \PP} r} \left[ R_{\text{memoryless}} - D \Delta_I \right] \approx \frac{\textup{Information Rate}}{\abs{X}} 
\end{split}
\end{align}

The neglected error terms in eq.~\eqref{eq:information-rate-approximation2} are $D^2$, $H\left(\mathcal{T}\mid Y\right)/ \abs{X}$ and $H\left( \mathcal{T} \mid X, Y \right)/ \abs{X}$.
The first term is negligible due to our assumption that $D \ll 1$, and we neglect the other two mainly because of the difficulty of including them in the optimization (the latter term is also neglected by Mitzenmacher and Drinea~\cite{drinea2007improved} and both are difficult to compute directly~\cite{kirsch2009directly}).

The last heuristic step in our analysis is to neglect the dependence of $\Delta_I$ on the input distribution $\PP$.
Instead, we will make estimates $\Delta \stackrel{?}{=} \Delta_I$, and maximize $R_{\text{heuristic}}$ assuming this value $\Delta$ of information loss per merge (but without limiting $\PP$ to distributions that maintain the equation $\Delta_I\left(\PP\right) = \Delta$).
This is a heuristic approximation, but it is somewhat justified assuming that $\Delta_I$ doesn't vary too wildly between otherwise ``good'' input distributions and that this variation is then multiplied by $D \ll 1$ in eq.~\eqref{eq:information-rate-approximation2}

Recall that the goal of eq.~\eqref{eq:information-rate-heuristic} is not to directly prove a lower bound, and may be far less accurate than using a calculation similar to the one described in~\cite{drinea2007improved}.
The main reason to use this approximation is that it gives a closed formula, and that a BAA-style convex optimization algorithm can be used to maximize it, which can then be used as parameters for Mitzenmacher and Drinea's construction.

\subsection{Reduction to Convex Optimization}\label{subsec:optimizing-the-heuristic-score-function}

In Section~\ref{subsec:analysis-of-the-heuristic}, we gave a heuristic argument for a simplified formula that can be used to estimate the information rate of Mitzenmacher-Drinea type codes for the binary deletion channel and the Poisson repeat channel.
This formula can be seen as an information rate formula with 2 correction terms corresponding to different ``costs'' of sending different run lengths (the cost due to the information lost in the event of a deletion and the overhead of sending longer runs).
In this section we will construct a Blahut-Arimoto type algorithm to maximize it.

In some sense, generalizing the Blahut-Arimoto algorithm to the cost types needed for our construction can be seen as an extension of the algorithm presented by Li and Cai~\cite{li2019blahut} who show how to extend the Blahut-Arimoto to quantum-classical codes where only distributions below a certain total cost are allowed.
We will reduce our optimization problem into an optimization problem in the class of problems solved by Li and Cai, and prove that our optimization yields a $\pm \varepsilon$-approximation of the optimal distribution score within $\textup{poly}\left(1/\varepsilon, 1/(1-d)\right)$ and $\textup{poly}(1/\varepsilon, \lambda, 1/\lambda)$ time for the binary deletion channel and the poisson repeat channel with parameters $d$ and $\lambda$ resp.

Our approximation for the rate of the resulting code is given by the formula
\begin{align}
    \begin{split}
        \label{eq:information-rate-heuristic2}
        R_{\text{heuristic}} &\defeq \frac{1}{\sum_r \PP_r r} \left[ R_{\text{memoryless}} - \sum_{r} \PP_r \Delta D_r \right]
    \end{split}
\end{align}

This approximation can be viewed as a combination of three terms:
\begin{itemize}
    \item The basic information rate of a memory-less channel of runs $R_{\text{memoryless}}$.
    \item A correction term for the information lost due to deletions $ \sum_{r} \PP_r \Delta D_r$.
    \item A ``price factor'' corresponding to the average resource cost of transmitting a run $\frac{1}{\sum_r \PP_r r}$.
\end{itemize}

The main claim we will prove in this section is that a nearly optimal distribution for this heuristic formula can be efficiently found:

\begin{lemma}
    \label{lem:polytime-optimization-algorithm}
    There exists an algorithm that returns a distribution $\PP$ for which \[R_{\textup{heuristic}} (\PP) \geq \sup_{\PP^\prime} \left\{R_{\textup{heuristic}} (\PP^\prime)\right\} - \varepsilon\] in $\poly\left(1 / \varepsilon, 1/(1-d), \Delta\right)$ time for the deletion channel with deletion probability $d$ and $\poly\left(1 / \varepsilon, \lambda, 1 / \lambda, \Delta \right)$ time for Poisson repeat channel with parameter $\lambda$.
\end{lemma}

\begin{proof}[Proof of Lemma~\ref{lem:polytime-optimization-algorithm}]
    The main difficulty in maximizing eq.~\eqref{eq:information-rate-heuristic2} is that the price factor term $\frac{1}{\sum_r \PP_r r}$ causes the relationship between $R_{\text{heuristic}}$ and $\PP$ to be non-convex.
    This issue will be the main focus of this section, and we will overcome it by extracting $L(\PP) \defeq \sum_r \PP_r r$ to be an external hyperparameter of the optimization.
    We then separate the search into many searches under the external condition the $L(\PP)=\ell$ for different values of $\ell$ and show that polynomially many searches suffice.

    Unlike $\Delta$, neglecting the effect of the fact that $\PP$ must satisfy $L(\PP) = \ell$ in each of these searches would lead us to select distributions of runs too heavily skewed towards long runs, lowering the rate they can achieve.
    So for any given value $\ell$, we limit our search to distributions $\PP$ under the condition that $L(\PP) = \ell$.
    The central point of our analysis will be to show that the supremum of $R_{\text{heuristic}}$ with respect to $\PP$ under the condition that $L(\PP) = \ell$, cannot change too rapidly with $\ell$, this is what allows us to enumerate over only polynomially many guesses for $\ell$ in order to approximate the optimal distribution $\PP$ to within a small additive error.

    Let $\eta \defeq \max\{\log{\Delta}, 1\}$.
    We begin with the simple observation that \[R_{\text{heuristic}} \left(\PP\right) \geq \Omega\left(\frac{1-d}{\eta}\right),\; \Omega\left(\frac{\min \{1, \lambda\}}{\eta}\right)\] can be achieved for the BDC and the PRC respectively.
    In particular, this is true for the Morse distribution \[\PP_i = \frac{1}{2} \delta_{i,t} + \frac{1}{2} \delta_{i,2t}\] which returns either $t$ or $2t$ each w.p. $1/2$, for \[t = \ceil{\eta\frac{100}{1-d}}, \ceil{\eta \max \left\{100, \frac{100}{\lambda}\right\}}\]
    This distribution obtains a non-negligible rate, because $L(\PP) \leq t = O\left(\frac{\eta}{1-d}\right)$, the deletion probability for a run from this distribution is $D < \frac{e^{-100}}{\Delta}$ and the probability of missing the reconstruction of any run is also very small, implying that \[R_{\text{heuristic}} \left(\PP\right) L(\PP) = R_{\text{memoryless}} (\PP) - \Delta D (\PP)  = \Omega(1)\]

    It is easy to see that setting $L(\PP) = \ell$ to be extremely large would result in an information transfer rate of at most $O(\log(\ell) / \ell) = o(1)$.
    Therefore, in order to approximately maximize $R_{\text{heuristic}} (\PP)$, the average lengths of runs in this distribution $L(\PP)$ is at most polynomially large in the parameters of Lemma~\ref{lem:polytime-optimization-algorithm}.

    Our next goal will be to prove that maximizing $R_{\text{heuristic}}$ for a (small) discrete subset of values of $\ell$ suffices to approximate its maximum on the whole range of run lengths.
    We denote by $h(\PP)$ the non-normalized version of eq.~\ref{eq:information-rate-heuristic2}:

    \begin{align}
            \begin{split}
                \label{eq:information-non-normalized}
                h(\PP) &\defeq R_{\text{heuristic}} \left(\PP\right) L(\PP) = R_{\text{memoryless}} - \Delta D \\
                &= \sum_{i, j} \PP_i P_{i\rightarrow j} \log \left( \frac{P_{i\rightarrow j}}{\sum_{i^\prime} \PP_{i^\prime} P_{i^\prime \rightarrow j}} \right) - \Delta \sum_i \PP_i D_i
            \end{split}
    \end{align}

    Let \[S_\ell = \left\{\PP \in \R^{\N} \mid \E_{r\leftarrow \PP} r = \ell \wedge\sum_i \PP_i = 1 \wedge \forall i\; \PP_i \geq 0 \right\}\] be the set of distributions with average cost $\ell$.
    We define $I(\ell) \defeq \sup_{\PP\in S_\ell} {h(\PP)}$ to be the optimal value of the non-normalized rate when fixing the average input run length to some value $\ell$.

    The main property of $I$ that we will use in our analysis is that it is monotonously non-decreasing in $\ell$ on the range $(0, \infty)$.
    \begin{claim}
        \label{claim:I-monotone}
        For any $0 < \ell_1 < \ell_2$, it holds that $I(\ell_2) \geq I(\ell_1)$.
    \end{claim}

    We leave the proof of Claim~\ref{claim:I-monotone} to the end of the section.
    We use it to bound from above the speed with which our approximation for the rate of the code can change when optimizing under slightly different $L=\ell$ constraints.
    In particular, the fact that $I$ is non-decreasing means that for any $\ell_1, \ell_2$ and for any $\ell\in \left(\ell_1, \ell_2\right)$, we have:

    \begin{equation}\label{eq:rate-of-change-of-rate}
        \begin{aligned}
            \frac{\ell_1}{\ell_2} \sup_{\PP_1 \in S_{\ell_1}} {R_{\text{heuristic}} \left(\PP_1\right)} &= \frac{I\left(\ell_1\right)}{\ell_2} < \sup_{\PP \in S_{\ell}} {R_{\text{heuristic}} \left(\PP\right)}= \\
            &= \frac{I\left(\ell\right)}{\ell} < \frac{I\left(\ell_2\right)}{\ell_1} = \frac{\ell_2}{\ell_1} \sup_{\PP_2 \in S_{\ell_2}} {R_{\text{heuristic}} \left(\PP_2\right)}
        \end{aligned}
    \end{equation}

    Eq.~\ref{eq:rate-of-change-of-rate} proves that it suffices to compute $I(\ell)$ only in a discrete set of points $\ell_i$, strictly separated on the logarithmic scale, in order to approximate it everywhere.
    Because we already showed that the value of $\ell$ for which $R_{\text{heuristic}}$ is maximized is polynomially bounded as a function of the parameters of Lemma~\ref{lem:polytime-optimization-algorithm}, this implies that polynomially many samples suffice to approximate this optimization.

    The last step in our construction is to show that $I(\ell)$ can be efficiently approximated for any given $\ell$.
    To this end we employ the algorithm proposed by Li and Cai~\cite{li2019blahut} who show that the Blahut-Arimoto algorithm can be extended to maximize functions of the form of $h(\PP)$, under ``cost limit'' constraints of the form $L(\PP) \leq \ell$.
    For any given $\ell$ we can generate a candidate distribution $\PP^\prime \in \bigcup_{\ell^\prime \leq \ell} S_\ell$ using Li and Cai's algorithm, and then convert it into a distribution $\PP\in S_\ell$ for which $h(\PP)$ is arbitrarily close to $h(\PP^\prime)$ using the same construction as in the proof of Claim~\ref{claim:I-monotone}.

\end{proof}

\begin{proof}[Proof of Claim~\ref{claim:I-monotone}]
    Let $\frac{1}{2} > \varepsilon > 0$ be some number, and let $\PP_1 \in S_{\ell_1}$ be some distribution for which $h(\PP_1) \geq I(\ell_1) - \varepsilon$.
    Our goal will be to show that there exists some $\PP_2 \in S_{\ell_2}$ for which $I(\ell_2) \geq h(\PP_2) \geq h(\PP_1) - \varepsilon \geq I(\ell_1) - 2 \varepsilon$.
    Because we prove this for an arbitrarily small $\varepsilon$, it will imply that $I(\ell_2) \geq I(\ell_1)$.

    Let $\PP_2 = \left(1-c \right) \PP_1 + c \delta_N$ be the distribution that returns a random sample from $\PP_1$ with probability $1-c$ and the value $N$ otherwise, where $c = \frac{\ell_2 - \ell_1}{N - \ell_1} > 0$, and \[N > \max \left\{ 2\ell_1, \frac{1}{\left(\ell_2 - \ell_1 \right) ^ 2}, \left(\frac{R_{\text{memoryless}}\left(\PP_1\right)}{10 \varepsilon}\right)^2, \left(\frac{1}{10 \varepsilon}\right)^2 \right\}\] is a sufficiently large integer.
    From its construction $\PP_2 \in S_{\ell_2}$.

    Our next goal is to show that $h(\PP_2) \geq h(\PP_1) - \varepsilon$.
    We do this by opening up the definition of $h$:

    \begin{equation}
        \label{eq:h-PP2)}
        h\left(\PP_2\right) = \underbrace{R_{\text{memoryless}} \left(\PP_2\right)}_{\geq (1 - c) R_{\text{memoryless}} \left(\PP_1\right)}
            - \underbrace{\Delta \E_{r\leftarrow \PP_2}{d^r}}_{\leq \Delta \E_{r\leftarrow \PP_1}{d^r} + c \Delta} \geq h\left(\PP_1\right) - \underbrace{c \left(\Delta + R_{\text{memoryless}}\left(\PP_1\right) \right)}_{\leq \varepsilon}
    \end{equation}
\end{proof}

	\bibliographystyle{alpha}
	\bibliography{main}

\end{document}